\documentclass[11pt]{article}
\usepackage{amsmath,amssymb,amsthm}
\usepackage[margin=1.0in]{geometry}
\usepackage{amscd}
\usepackage{amsthm}
\usepackage{epsfig}
\usepackage{url}
\allowdisplaybreaks
\makeindex
\begin{document}
\title{\bf {PURE DATA SPACES}}
\author{%
  SAUL YOUSSEF%
  \hfil \\
  Department of Physics \\
  Boston University \\
  Saul.Youssef@gmail.com
}
\maketitle
\begin{abstract}
In a previous work, {\it pure data} is proposed as an axiomatic foundation for mathematics and computing, based on {\it finite sequence} as 
the foundational concept rather than based on {\it logic} or {\it type}.   Within this framework, objects with mathematical meaning are 
{\it data}, and, we argue, collections of mathematical objects must then be associative data, called a {\it space}.  A space is then the basic collection in this 
framework analogous to sets in Set Theory or objects in Category Theory.  A theory of spaces is developed, where 
spaces are studied via their semiring of endomorphisms.  To illustrate these concepts, and as a way of 
exploring the implications of the framework, pure data spaces are ``grown organically'' from the substrate of 
pure data with minimal combinatoric definitions.  Familiar objects from classical mathematics emerge this way,   
including natural numbers, integers, rational numbers, boolean spaces, matrix algebras, Gaussian Integers, Quaternions, 
and non-associative algebras like the Integer Octonions and Sedenions. 
Insights from these examples are discussed with a view towards new directions in theory and new exploration.
\end{abstract}
%%%%%%%%%%%%%%%%%%%%%%%%%%%%%%%%%%%%%%%%%%%%%%%%%%%%%%%%%%%%%%%%%%%%%%%%

%%%%%%%%%%%%%%%%%%%%%%%%%%%%%%%%%%%%%%%%%%%%%%%%%%%%%%%%%%%%%%%%%%%%%%%%
\theoremstyle{definition}
\newtheorem{axiom}{Axiom}
\newtheorem*{axiom*}{Axiom}
\newtheorem*{fact}{Fact}
\newtheorem{theorem}{Theorem}[section]
\newtheorem{lemma}{Lemma}
\newtheorem{corollary}[theorem]{Corollary}

\newtheorem{definition}{Definition}

\newtheorem*{remark}{}

\section{Introduction}

    In {\it Pure Data Foundation of Mathematics and Computing} \cite{PDF}, we proposed a formal system based on {\it finite sequence} as the foundational concept, rather than based on {\it logic}, as in Set Theory or 
based on {\it type}, as in dependent type theories.  In this framework, pure data acquires meaning via definitions.  
A collection of definitions defines equality of data which, therefore, relates mathematical quantities to 
each other.  This makes for a drastically simplified relationship between proof and computation.   
Any sequence $A_1$=$A_2$=$\dots$=$A_n$ is both a valid proof of $A_n$ given $A_1$ and a valid computation 
of $A_n$ starting with $A_1$.  

       The previous work mainly developed  
consistency, logic, and the resolution of paradoxes\cite{PDF}.  A goal of this work is to begin the exploration of classical 
mathematics through the lens of pure data.  In particular, doing mathematics requires 
collections of objects as in Set Theory or Category Theory.  In the pure data context, minimal requirements for 
data to contain a collection of other data leads to a simple conceptual definition of a ``space.''  A space is any data 
which is ``associative.''  Although a space is like a set in the sense that it is a minimal collection, we surprisingly 
find that spaces have a rich internal structure which is shared by everything from a single atom to the entire 
space of pure data.  Understanding spaces appears to be the key 
to better understanding the relationship of this framework to classical mathematics.  

     In Sections 1-4, we review the axiomatic framework and establish a suitable minimal set of starting definitions.  
 Section 6 develops the theory of spaces, morphisms of spaces, semirings, semialgebras and fields.   Sections 7-14 
 are an exploration of the spaces which ``organically'' emerge from pure data, and how they relate to familiar mathematical objects.  
 Section 15 summarizes and suggests research directions. 

\section{Foundation}  

In \cite{PDF}, we argue that mathematics and computing can both be captured within a small axiomatic framework, 
based on {\bf finite sequence} as the foundational concept.  Assuming that finite sequences are understood, {\it data} and {\it coda} are defined by 

\begin{definition} {{\bf Data} is a finite sequence of {\bf codas}, where each {\bf coda} is a pair of {\bf data}.}
\end{definition}

\noindent The two foundational operations defined on data are 1) concatenation of data $A$ and data $B$, written `$A\ B$', and 2) pairing of data $A$ and data $B$ as a coda, written with a colon `$A:B$'.   
By Definition 1, the empty sequence, written `()', qualifies as data and, therefore, () paired with itself is a coda, written `(():())' or `(:)'.  
Any finite sequence of codas is data, so, for example, (:)\ (:(:))\ ((:):(:(:))) is  
data consisting of a sequence of three codas.  We can think of this as {\it pure data} since it is ``data made of nothing.''  
By convention, the colon binds from the right first and binds less strongly than concatenation, so that $A:B:C$ is defined to be $(A:(B:C))$ and $A:B\ C$ is defined to be $(A:(B\ C))$.  
Data is typically written with upper case, and codas are typically written in lower case.  
To indicate the left and right data of a coda, we sometimes use L/R superscripts so, for any coda $c$, $c=(c^L:c^R)$. 

     All meaning within the system is determined by a chosen partial function from coda to data called a {\it context}. 
Contexts are partially ordered by inclusion, so if $\delta$ and $\delta'$ are contexts, $\delta\le\delta'$ if $\delta$ and $\delta'$ are equal on the domain of $\delta$.  
Given a context $\delta$, equality of data is the equivalence generated by $c\sim \delta(c)$ for any coda $c$, and by compatibility with concatenation and colon in the sense that 
if $A\sim B$, then $(A\ X)\sim (B\ X)$, $(X\ A)\sim (X\ B)$, $(A:X)\sim (B:X)$ and $(X:A)\sim(X:B)$ for any data $X$.  This relation $A\sim B$ is denoted $A{\overset\delta =}B$, 
or simply $A=B$ when the context is unambiguous.   
It follows that if $\delta\leq\delta'$, then $A{\overset\delta =}B$ implies $A{\overset{\delta'} =}B$.  We say that if $A$ and $B$ are equal then they are ``always equal,'' thinking 
of moving from $\delta$ to $\delta'$ as the passage of ``time."  The fixed points of a context $\delta$ are {\it atoms}.  Note that if $c$ is an atom in context $\delta$, 
and if $\delta\leq\delta'$, then $c$ is still an atom a context in $\delta'$.  Thus, atoms are ``permanent" and if $c$ is an atom, $c$ is ``always" an atom.  

\begin{definition}
{Within a given context $\delta$, coda $c$ is a {\bf fixed point} if $\delta:c\mapsto c$.  A coda equal to a fixed point is an {\bf atom}.  If data $A$ contains an atom in its sequence, $A$ is {\bf atomic} data. 
Data $A$ is {\bf invariant} if every coda $a$ in the sequence of $A$ is a fixed point and if $a^L$ and $a^R$ are both {\bf invariant}. }
\end{definition}
\noindent It follows from the definition that empty data is invariant and atoms $a$ and $b$ are equal if and only if $a^L=b^L$ and $a^R=b^R$.  If $a$ and $b$ are atoms, 
$A$ and $B$ are data, then $(a\ A)=(b\ B)$ and $(A\ a)=(B\ b)$, if and only if $a=b$ and $A=B$.  If $A=B$ and $A$ and $B$ are invariant, then $A$ and 
$B$ are identical as pure data.

\section{Genesis}

    In the proposed framework, all mathematical objects are pure data and all meaning is defined by a chosen context $\delta$ and its corresponding equivalence.  
A {\bf proof} in $\delta$, for example, is just a sequence $A_1{\overset \delta =}$ $A_2 {\overset \delta =} \dots {\overset \delta =}$$A_n$ which can be viewed either as proving $A_n$ 
given $A_1$ or as a computation $A_1\mapsto A_n$. 
The immediate issue is how to choose $\delta$.  
    
     In the beginning, there are no definitions, and the corresponding context is the empty partial function $\delta_0$.  
Since the domain of $\delta_0$ is empty, it has no fixed points and, therefore, no atoms.  Thus, the empty sequence is the only invariant, 
$X{\overset \delta =}X$ is the only valid proof, and the only valid 
 computations do nothing $X\mapsto X$ for any data $X$.
To define a non-empty context $\delta_0\le\delta$, we need, at least, an invariant specification of the domain of $\delta$ as a partial function.    
Since the empty sequence is the only invariant, the only way to specify if a coda (A:B) in the domain of $\delta$ is 
to require either A or B or both to be the empty. In each of these cases, the coda (:) is within the domain of $\delta$, and so we must, in any case, 
decide on what (:) maps to.  There are three possibilities; 
\begin{itemize}
\item[] choice 1: {$\delta: (:) \mapsto ()$},
\item[] choice 2: {$\delta: (:) \mapsto (:)$}, or 
\item[] choice 3: {$\delta: (:) \mapsto \ anything\ other\ than\ ()\ or\ (:)$}. 
\end{itemize}
Since pure data is ``made of (:)'', choice 1 trivially causes all data to be equal to the empty sequence. 
On the other hand, choice 3 would mean that for any non-empty data $A$, the number of (:) atoms in $A$ grows without limit.  This would mean, for instance, 
that no computation could produce a final result.  Thus, we are constrained to choice 2, where (:) is an invariant atom in $\delta$, and, therefore an atom in 
any context $\delta\leq\delta'$ as well.  It is convenient to generalize choice 2 and to let $\delta$:(:$X$)$\mapsto$ (:$X$) for all data $X$, since this provides a 
supply of invariant atoms (:), (:(:)), (:(:) (:)),$\dots$ which can be used to expand the domain of $\delta$.  

We adopt a convention for expanding the domain of $\delta$ while preserving context partial ordering. 
A context of the form 
\begin{equation}
	\delta_a: (a\ A:B) \mapsto \delta_a(A,B)
\end{equation}
for some invariant atom $a$ is called a {\bf definition}.  We conventionally restrict ourselves to 
contexts $\delta\cup\delta_a\cup\delta_b\cup\dots$ where $a$ is an invariant atom in $\delta$, $b$ is an invariant atom in $\delta\cup\delta_a$, $c$ is an invariant atom in $\delta\cup\delta_a\cup\delta_b$ and so forth.  By requiring $a$, $b$,\dots to be disjoint, we maintain the context partial ordering.  

\section{Definitions}

     Before proceeding to mathematical exploration, we need to add enough expressive power to the original context $\delta$.  To do this, 
we introduce minimal definitions using the mechanism explained in the previous section.  

\begin{enumerate}
\item {Define ``bits'', ``bytes'' and ``byte sequences'' so that words are single atoms, so that $a$ can be a word in future definitions $\delta_a$.}
\item {Define all of the naturally occurring combinatoric operations given two finite sequences $A$ and $B$.  For example $\delta_{\rm ap}$
maps (ap $A$ : b$_1$\ b$_2$ $\dots$ b$_n$) to ($A$:b$_1$)\ ($A$:b$_2$)$\dots$ ($A$:b$_n$), so that ap ``applies $A$ to each atom of input.'' }
\item {Define a minimal internal language $\delta_{\{\}}$ so that a coda $(\{\texttt{language expression}\}\ A : B)$ is mapped to corresponding 
data as described in Section 4.3 below.}
\item {Include definitions giving direct access to context-level values.  This includes a Boolean definition $\delta_{\rm bool}$ to determine if data is empty, 
$\delta_{=}$ to test if two data are equal in $\delta$, and $\delta_{\rm def}$ to add definitions via Equation 1.} 
\end{enumerate}
These are meant to be a minimal collection of definitions only including what inevitably comes from the combinatorics of the elements of 
Equation 1 as finite sequences. 
There are intentionally no definitions which presume underlying arithmetical operations.  A natural number, 
`123', for instance is a single atom byte sequence and there is no definition which implies interpreting this as a number in the usual way.  The idea is
to allow natural numbers and whatever else to emerge ``organically.'' 
There are approximately fifty definitions needed to create a basic system.  We define some here for orientation.  The remaining definitions used 
in the text can be found in the Glossary.

\subsection{Bits, Bytes and Byte sequences}

If `a' is an invariant atom, a new atom with domain (a A:B) is defined by $\delta_a: (a\ A:B)\mapsto (a\ A:B)$, so that $\delta_a$ is a fixed point on its domain.  
We may refer to these as ``a-atoms."  For readability convenience, we define (:)-atoms, ((:):)-atoms and ((:):(:))-atoms to hold bits, bytes and words respectively
so that text strings are invariant atoms within the context $\delta\cup\delta_{(:)}\cup\delta_{((:):)}\cup\delta_{((:):(:))}$.

\subsection{Combinatorics} 

Given that text strings are invariant atoms, we can add definitions with readable names.  Define the identity `pass' so that (${\rm pass}:X$=$X$) for all data $X$.  
\begin{itemize}
\item{$\delta_{\rm pass}$ : ({\rm pass} A:B) $\mapsto$ B}
\end{itemize}
and define `null' so that (${\rm null}:X$=$()$) for all $X$.
\begin{itemize}
\item{$\delta_{\rm null}$ : ({\rm null} A:B) $\mapsto$ ()}
\end{itemize}
It is convenient to define ``getting the A and B components of a coda'' like so
\begin{itemize}
\item{$\delta_{\rm left}$ : ({\rm left} A:B) $\mapsto$ A} 
\item{$\delta_{\rm right}$ : ({\rm right} A:B) $\mapsto$ B} 
\end{itemize}
We proceed to define the minimal combinatorics involving of A and B as finite sequences.  These don't have 
commonly understood names, so we introduce new names.  For example the name `ap' is meant to suggest the 
concept ``apply A to each atom of B" is expressed by the definition.  
\begin{itemize}
\item{$\delta_{\rm ap}$ : ({\rm ap} A : b B) $\mapsto$ (A:b) ({\rm ap} A : B)}
\item{$\delta_{\rm ap}$ : ({\rm ap} A : ()) $\mapsto$ ()} 
\end{itemize} 
The domain of $\delta_{\rm ap}$ is defined with the assumption that $b$ is an atom, so the domain of the first branch of $\delta_{\rm ap}$ is exactly 
codas $({\rm ap}\ A,B)$ where $B$ starts with an atom.  By definition, if we list multiple ``branches'' as in this case, the first branch in the domain applies.  

\subsection{Language} 

Axiomatic systems like ZFC\cite{ZFC} or dependent type theories\cite{types,hott} are formal languages.  There is an alphabet, special symbols, and syntax rules 
for defining valid expressions.  Our approach avoids this by defining a language as just one more definition like any other.  The basic idea is to 
define minimalist language, mainly giving access to the two foundational operations via text blank space (for concatenation) and text colon (for pairing to data as a coda).  So, the idea is 
\begin{itemize}
\item {$\delta_{\{\}}$ : (\{x\ y\} A : B) $\mapsto$ (\{x\} A:B) (\{y\} A:B)}
\item {$\delta_{\{\}}$ : (\{x:y\} A : B) $\mapsto$ (\{x\} A:B) : (\{y\} A:B)}
\end{itemize}
where x and y are byte sequences as defined above.  As written, this is ambiguous since $x$ and $y$ may contain space and colon characters, but 
these ambiguities can be resolved just by choosing and order, thus forcing the language into one definition 
$\delta_{\{\}}$.  This definition includes 
\begin{itemize}
\item {$\delta_{\{\}}$ : (\{A\} A : B) $\mapsto$ A}
\item {$\delta_{\{\}}$ : (\{B\} A : B) $\mapsto$ B}
\end{itemize}
so that A and B have special meaning in the language.  As a result, for example, (\{B\}:1 2 3)=1 2 3, (\{B B\}:1 2 3)=1 2 3 1 2 3, and (\{A B\} a b : 1 2)=a b 1 2.  

     Given a language expression in the form of byte sequence data $L$, the corresponding data is, simply, ($L$:).  Every sequence of bytes is a valid language 
expression, so there is actually no need to define or check for syntax errors.  Similarly, an {\it evaluation} of ($L$:) is merely some sequence 
($L$:)=$D_1$=$D_2$=$\dots$=$D_n$ producing an ``answer'' $D_n$.  Typically, this is done with a simple strategy based on applying $\delta$ whenever 
possible or up to time or space limits.  This also means that there is also no such thing as a ``run time error.''    

\subsection{System} 

    The most basic question to ask about data is whether it is empty or not.  This is captured by the definition $\delta_{\rm bool}$, defined 
so that (bool:B) is () if B is empty (``true'') and (bool:B)=(:) if B is atomic (``false'').
\begin{itemize}
\item {$\delta_{\rm bool}$ : ({\rm bool} A : B) $\mapsto$ () if $B$ is empty, (:) if $B$ is atomic.}
\item {$\delta_{\rm not}$ : ({\rm not} A : B) $\mapsto$ (:) if $B$ is empty, () if $B$ is atomic.} 
\end{itemize} 
The equality defined by the context $\delta$ is available within $\delta$ via the following definition:
\begin{itemize}
\item {$\delta_{=}$ : (= A : ()) $\mapsto$ A}
\item {$\delta_{=}$ : (= (): A) $\mapsto$ A}
\end{itemize}
and if a and b are atoms, 
\begin{itemize}
\item {$\delta_{=}$ : (= a A : b B) $\mapsto$ (= a:b) (=A:B)}
\item {$\delta_{=}$ : (= A a : B b) $\mapsto$ (=A:B) (= a:b)}
\end{itemize}
so, if $(A:B)$ is empty, $A$ and $B$ are ``always'' equal and if $(A:B)$ is atomic, $A$ and $B$ are ``never'' equal.  The full language has a little bit 
of syntactic sugar, so one can write (A=B) instead of (= A:B) [ref]. 

New definitions can also be added to context via  
\begin{itemize}
\item {$\delta_{def}$ : ({\rm def} \texttt{name} A : B} $\mapsto$ ())
\end{itemize}
which is defined to be in domain as a partial function if \texttt{name} is not already in the current context, and, in that case, it adds the following definition
to context.
\begin{itemize}
\item {$\delta_{\texttt{name}}$:(\texttt{name}\ $A'$:$B'$) $\mapsto$ ($B\ A'$:$B'$)}
\end{itemize}
so, for instance (def first2 : {first 2 : B}), means that (first2 : a b c d) is interpreted as (\{first 2:B\} : a b c d) which is equal to the data (a b).   

We take this as the ``organic'' base of naturally occurring 
definitions plus the language will be a starting point in searching for the ``spaces'' defined in section 5.  Further definitions used in the text can be found in 
the Glossary.  Examples, tutorials, a complete definition of the language, and software can be 
found in reference \cite{coda}. 

\section{Global Structure} 

    Pure data has a global structure in the sense that the foundational operations $(A\ B)$ and $(A:B)$ are defined for any data $A$, $B$.  
This suggests defining a corresponding associative product $(A\cdot B)$ and associative sum $(A+B)$ by    
\begin{equation}
(A \cdot B):X = A:B:X 
\end{equation}
\begin{equation}
(A+B):X = (A:X)\ (B:X) 
\end{equation}
for all data $X$. This product distributes over the sum, from the right only, so we have 
\begin{equation}
(A+B)\cdot C=(A\cdot C)+(B\cdot C) 
\end{equation}
for any data $A$, $B$, and $C$.
Since $({\rm pass}\cdot A)$=$(A\cdot {\rm pass})$=$A$ and $(A+{\rm null})$=$({\rm null}+A)$=$A$ for any $A$, (pass) and (null) are the units of data composition and 
data addition respectively.  
Data $A$ is {\it idempotent} if $A\cdot A$=$A$, is an {\it involution} if $A\cdot A$=$1$, and $A$ has an {\it inverse} if $A\cdot A'$=$A'\cdot A$=$1$ for some data $A'$.   
A product ($A_n\cdot A_{n-1}\dots A_2\cdot A_1$) {\it starts with} $A_1$ and {\it ends with} $A_n$.   
Data $A$ is {\it algebraic} or {\it commutative} if $A:X\ Y$=$A:Y\ X$ for all $X$, $Y$, and is 
{\it distributive} if $A:X\ Y$=$(A:X)\ (A:Y)$ for all $X$ and for all $Y$.  
Each data $A$ defines a binary operator on data via ($X{\overset A\ast}Y)\mapsto (A:X\ Y)$.  Since ${\overset A\ast}$ is associative if and only if 
 ($A$:$X$\ $Y$)=($A$:($A$:$X$)\ $Y$)=($A$:$X$ ($A$:$Y$)) for all $X$,$Y$ we say that data $A$ is {\it associative} if it has this property.  
 
Algebraic and distributive data foreshadow a theme where these two properties make data mathematically interesting for complementary reasons.  
Algebraic data has, in a sense, ``transcended finite sequence," the foundational concept of the system, and deserves a platonic existence in 
that sense.  Distributive data, on the other hand, is maximally sequence dependent, but it may also be mathematically interesting because  
such data are ``morphisms of finite sequence.''   
 
\section{Spaces}

      A basic requirement for useful mathematics is to have a way to define and refer to collections.  
Within the framework of pure data, however, there is only one ``substance'' to work with.  All mathematical objects are data and 
all collections of mathematical objects are also data.  
Thus, we are lead to ask, given some data $S$, how could $S$ represent a collection of other data?  
There are a couple of plausible ways to do this.  
\begin{itemize} 
\item The collection corresponding to $S$ could be the collection of data ($S$:$X$) for any data $X$;
\item The collection corresponding to $S$ could be the collection of fixed points of $S$.  
\end{itemize}
These ideas coincide if we require $S$ to be idempotent, so that ($S$:$X$) is automatically a fixed point.  It is natural to also expect compatibility with 
sequences in the sense that if ($S$:$X$) is in the collection and ($S$:$Y$) is in the collection, then ($S$:$X$) ($S$:$Y$) is also in the collection.  
This is guaranteed if we require 
\begin{equation}
(S : X\ Y) = S : (S:X)\ (S:Y)
\end{equation}
for all data $X$, $Y$.  
Note that data $S$ is idempotent and satisfies Equation 5 if and only if $S$ is associative, and so we are lead 
to a simple definition.
\begin{definition} Data $S$ is a {\bf space} if $S$ is associative.
\end{definition}
\noindent If $S$ is a space, we say that any data ($S$:$X$) is {\it in} $S$.  
The data ($S$:) is called the {\it neutral data} of $S$.  

Examples of spaces can be found in the basic definitions, including `pass', `null' and `bool'.  More examples come from 
noting that any data which is both idempotent and distributive is a space.  Thus, if $J$ is idempotent, then (ap\ $J$) is idempotent and 
distributive, and is, therefore, a space.  If spaces $S$ and $T$ commute, then $S\cdot T$ is a space.    

\subsection{Morphisms} 

A product $F$ that starts with space $S$ and ends with space $T$ is a {\it morphism} from $S$ to $T$ and can be written $S{\overset F\longrightarrow}T$.
Since spaces are idempotent, $F$=($F\cdot S$)=($T\cdot F$).  Since $T\cdot S$ is a morphism from $S$ to $T$, morphisms always exist.  
Since the product is associative, morphisms can be composed, so that if $S{\overset F\longrightarrow}T$ and 
$T{\overset G\longrightarrow}U$, then $S{\overset {G\cdot F}\longrightarrow}U$.  
If $F$ and $G$ are morphisms from $S$ to $T$, we can also define a {\it sum of morphisms} by $F \oplus G = T\cdot (F+G)$, so that 
$S{\overset {F\oplus G}\longrightarrow}T$.  A morphism from a space to itself is an {\it endomorphism}. 

    Since $T\cdot X\cdot S$ is a morphism from $S$ to $T$ for any data $X$, a morphism can define any function.  Special morphisms which 
``preserve the structure'' of $S$ and $T$ are defined as follows.  A morphism $S{\overset H\longrightarrow}T$ is a {\it homomorphism} if 
\begin{equation}
H\cdot (f\oplus_S g) = (H\cdot f) \oplus_T (H\cdot g) 
\end{equation}
for all endomorphisms $f$ and $g$ of $S$.  Compare with morphisms that happen to also be a space.  Morphism $S{\overset U\longrightarrow}T$ is a space if and only if 
\begin{equation}
U\cdot (f\oplus_S g) = U\cdot ((U\cdot f) \oplus_T (U\cdot g)) 
\end{equation}
for all endomorphisms $f$ and $g$ of $S$.  Thus, any idempotent homomorphism is a space.  
Note that homomorphisms are closed under composition, but the composition of two spaces is guaranteed to be a space 
only if the spaces commute.  

\subsection{Spaces are Semirings} 

     Consider the endomorphisms of a fixed space $S$.  Since $S$ itself is a product starting at $S$ and ending at $S$, $S$ qualifies as an endomorphism of itself.  	
 Since ($S\cdot f$)=($f\cdot S$)=$f$ for any endomorphism $f$, $S$ is ``its own identity morphism.''  The endomorphisms 
of $S$ are closed under both composition and addition, with addition defined by $f\oplus g$ = $S\cdot(f+g)\cdot S$ = $S\cdot(f+g)$.  
Thus, the endomorphisms of $S$ are a `semiring' with composition and addition and with $1$=($S\cdot {\rm pass}\cdot S$)=$S$ and with 
$0$=($S\cdot {\rm null}\cdot S$)=($S\cdot {\rm null}$).  

\begin{definition}
A {\bf semiring} is a set with associative addition and multiplication (denoted $f\oplus g$ and $f\cdot g$), with additive and 
multiplicative identities {\bf 0} and {\bf 1} satisfying $f\oplus 0$=$f$=$0\oplus f$ and $f\cdot 1$=$f$=$1\cdot f$ for all $f$ in the semiring, and where 
$(f\oplus g)\cdot h$=$(f\cdot h)\oplus (g\cdot h)$ for all $f$, $g$, and $h$ in the semiring.  
A subset of a semiring $S$ which contains 0 and 1 and is closed under addition and multiplication is a {\bf subsemiring} of $S$.   
\end{definition} 
\noindent 
Note that the semiring of space $S$ has 1=0 if and only if $S$ is a constant space equal to (const $K$) for some data $K$.  
It is easy to verify that the semiring of the space (pass) is exactly the the global algebra defined in Section 5.  Since isomorphic spaces have 
isomorphic semirings, understanding the semiring of a space is mainly what we aim for in the examples which follow.  Before going to examples, however, it is helpful 
to identify semiring features of interest. 

Consider the semiring of a fixed space $S$.  
\begin{itemize}

\item{{\bf Subspaces.} If an endomorphism $s$ of $S$ happens to also be a space, we say that $s$ is a {\it subspace} of $S$.  The name ``subspace'' is justified because every data contained in $s$ is also contained in $S$, and every endomorphism $s\cdot$X$\cdot s$ of $s$ is also an endomorphism of $S$.  
Subspaces partially distribute from the left as with $s\cdot(f\oplus g)$=$s\cdot((s\cdot f)\oplus(s\cdot g))$.
Every space has subspaces $1$ (the whole space), and $0$, the constant neutral space.  If subspace $s$ is not equal to $S$, $s$ is a {\it proper} subspace.} 

\item{{\bf Constants.}  An endomorphism $k$ of $S$ satisfying $k\cdot 0$=$k$ is a {\bf constant} subspace of $S$.  There is one such constant for each 
data in $S$.  Since $k\cdot f$ and $k\oplus l$ are constants for 
constant $k$ and $l$ and for any endomorphism $f$, the constants of $S$ are a left-ideal of the semiring of $S$.} 

\item{{\bf Homomorphisms.}  By Equation 6, the homomorphisms of $S$ are exactly the endomorphisms which distribute over addition from
both the left and from the right.  The endomorphisms 0 and 1 are homomorphisms.  If $f$ and $g$ are homomorphisms of $S$, then 
$f\cdot g$ is a homomorphism.  If $S$ is algebraic, the homomorphisms are a subsemiring of the semiring of $S$.} 

\item {{\bf Group of units}. The collection of endomorphisms with multiplicative inverses is called the {\it group of units} of the space $S$.} 

\item{{\bf Central endomorphisms}.  Endomorphisms which commute with the group of units are {\it central endomorphisms}.  The endomorphisms 0 and 1 
are central.  If $f$ and $g$ are central endomorphisms, then $f\cdot g$ is a central endomorphism.  If $f$ and $g$ are {\it central homomorphisms}, then 
$f\oplus g$ is also a central homomorphism.  Thus, the central homomorphisms of $S$ are a subsemiring of the semiring of $S$.}

\item{{\bf Idempotents}.  Define a relation ${\overset i \leq}$ on the endomorphisms of $S$ by $f{\overset i \leq}g$ if ($f\cdot g$)=$f$.  The equivalence 
classes of the corresponding symmetric relation $f{\overset i \sim} g$ contain all the idempotent endomorphisms, since $f{\overset i\sim}g$ implies that 
both $f$ and $g$ are idempotent.  In this {\it idempotent order}, 1 is the maximum idempotent and the constant idempotents are minimal and equivalent to each other.}

\item{{\bf Neutral spaces}. If addition in $S$ allows left and right cancellation ($f\oplus g=f\oplus h$ implies $g=h$ and $g\oplus f=h\oplus f$ implies $g=h$), 
$S$ is a {\it neutral} space.  Neutrality of $S$ guarantees $h\cdot 0$=0 for homomorphism $h$, and guarantees that if $f\oplus f$=$f$, then $f=0$.} 

\item{{\bf Semilattice spaces}.  If $S$ is algebraic and $f\oplus f$=$f$ for all endomorphisms $f$, then $S$ is called a {\it semilattice space}.}

\end{itemize} 
In the case where $S$ is distributive, $(f\oplus g)$ = $(f+g)$ for $f$, $g$ endomorphisms of $S$.  This means that the homomorphisms of $S$ are 
just the distributive endomorphisms of $S$.  If the group of units of $S$ includes permutations, then the central homomorphisms are algebraic.  This illustrates 
the generalization of the earlier remark about algebraic and distributive data being mathematically interesting for complementary reasons.  Central 
endomorphisms are the generalization of algebraic data and homomorphisms are the generalization of distributive data.  
This suggests that the central homomorphisms of a space will be maximally mathematically meaningful.

\subsection{Isomorphic Spaces have Isomorphic Semirings} 

The fact that homomorphisms compose suggests adapting definitions from Category Theory.  
We say, for instance, that spaces $S$ and $T$ are {\it isomorphic} if the diagram of Figure 1 commutes for homomorphisms 
$h$ and $h'$.  
\begin{figure}[h]
\centering
\includegraphics[width=1.0\textwidth]{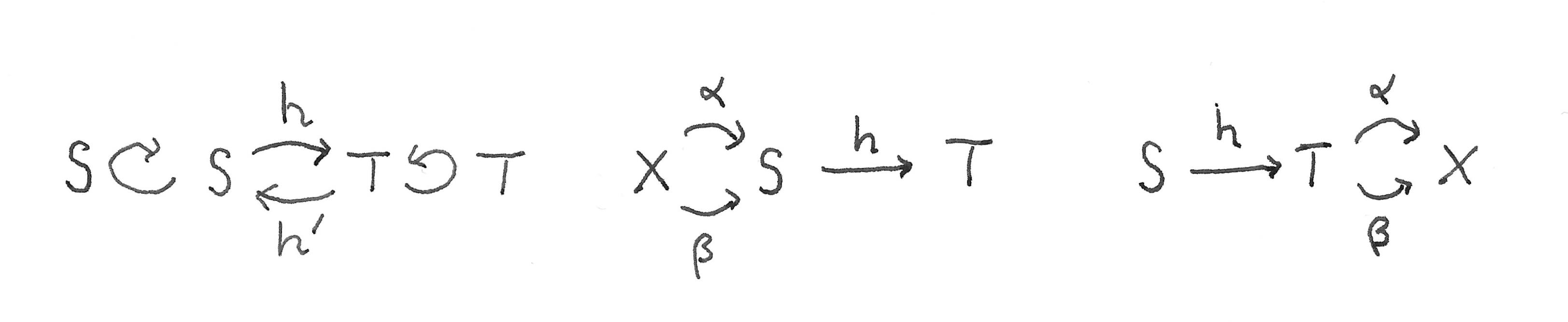}
\caption{{\it Since homomorphisms compose, standard definitions can be adopted from Category Theory where each space is its own ``identity morphism.'' 
$S$ and $T$ are isomorphic if the diagram on the left commutes for homomorphisms $h$ and $h'$.  A homomorphism $h$ is a monomorphism if the center 
diagram implies $\alpha=\beta$ and $h$ is an epimorphism if the right diagram implies $\alpha=\beta$.}}  
\end{figure}
where we have used $S$ and $T$ instead of respective ``identity morphisms."  Given an isomorphism as in Figure 1, the bijection $f\mapsto h'\cdot f \cdot h$ is 
a semiring isomorphism, preserving product, sum, 1 and 0, so that isomorphic spaces have the same features defined in the previous section, since 
these are all defined with semiring operations.  A homomorphism $h$ where $h\cdot \alpha=h\cdot \beta$ implies $\alpha=\beta$ is a {\it monomorphism}, and if $\alpha\cdot h=\beta\cdot h$ implies $\alpha=\beta$, then $h$ is an {\it epimorphism} \cite{categories}.  

\subsection{Spaces can be Semialgebras} 

The semiring of a space $S$ is close to the structure of an ``algebra'' in standard mathematics.  There is a binary operation $f\oplus g$, so 
constants of the space can be added.  Homomorphisms have the property that they map constants to constants and are distributive with 
$h\cdot (k\oplus l)$ equal to $(h\cdot k)\oplus (h\cdot l)$, so $S$ is almost an algebra with a distributive product.  
This suggests a definition. 

\begin{definition}
A space with a chosen mapping from constants to homomorphisms is a {\bf semialgebra}.  If the homomorphisms 
in the image of the mapping are central, the semialgebra is a {\bf central semialgebra}.  
\end{definition}
\noindent Since pure data is meant to include all mathematical objects, we might expect that important algebras of standard 
mathematics may appear as semialgebras or central semialgebras. 

\subsection{Subspaces are Substructures} 

Suppose that space $S$ has subspace $s$ and space $T$ has subspace $t$, and suppose that $F$=$t\cdot X\cdot s$ is a morphism from $s$ to $t$.   
Since ($t\cdot X\cdot s$) = ($T\cdot t \cdot X\cdot s\cdot S$), 
$F$ is also a morphism from $S$ to $T$, with the extra property that $F$ respects the structure of $s$ and $t$ in the sense that $t\cdot F=F\cdot s$.  
If a space has multiple subspaces, each commuting pair of subspaces is a new subspace with compatible structure.  

     Consider a subspace $s$ of a space $S$.  If a homomorphism $h$ of $S$ is also an endomorphism of $s$, then $h$ is also a 
homomorphism of the subspace $s$.  Similarly, if a unit $u$ of $S$ is an endomorphisms of $s$, then it is also a unit of $s$.  Such 
homomorphisms and units {\it descend} into the subspace $s$.  

       For any morphism $S{\overset h \longrightarrow}T$, there is a guaranteed factorization  
\begin{equation}
h = m \cdot e
\end{equation}
where $m$ is a monomorphism and $e$ is an idempotent of $S$.  If $h$ is a homomorphism, we can substitute this into the 
definition of a homomorphism, use $h\cdot e=h$, and, cancelling $m$ on the left, find 
\begin{equation}
e\cdot (f\oplus g) = e\cdot (e\cdot f \oplus e\cdot g),
\end{equation}
and conclude that $e$ is a subspace of $S$, which we call a {\it quotient} of $S$, and denote $e$ by $S/h$.  Note that $S/h=S$ if and only if $h$ is a unit as well as a homomorphism. 

\subsection{Fields}

A descending sequence of subspaces can only end with a constant space, since constant spaces are the only spaces with 
no proper subspaces.  Constant spaces have no variety - they are all isomorphic. However, there is an interesting 
variety one step above constants. These are the `fields.' 
\begin{definition}
A space where all proper subspaces are constants is a {\bf field}.
\end{definition}
\noindent In classical ring theory\cite{ring}, a field guarantees multiplicative inverses.  
An analogous result for spaces nicely combines homomorphisms, constants, units and subspaces.  
\begin{theorem}
A space $S$ is a field if and only if every non-constant homomorphism of $S$ is a unit. 
\end{theorem}

\begin{proof} If $S$ is a field, every non-constant homomorphism $h$ has $S/h$=$S$ and, thus, is a unit.  Conversely, 
suppose that $S$ is not a field.  Then $S$ contains a proper non-constant subspace $s$ and $\sigma = S\cdot {\rm ap}\ s\cdot S$ is,
therefore, a non-constant homomorphism of $S$.  Since $\sigma:X$ is in $s$ for any data $X$, $\sigma$ is not a unit. 
\end{proof}
\noindent 
The obvious examples of fields are null, bool and any constant space.  In the examples, we will see familiar fields such 
as the natural numbers modulo a prime or the field of rational numbers.   

\section{Plan}

Although the motivation for the concept of a space is merely to define a general collection of data, and a space is merely any associative data, we 
find that spaces and their endomorphisms have a rich internal structure including subspaces, a group of ``symmetries,'' a class of endomorphisms which 
qualify as homomorphisms, and central endomorphisms which commute with the symmetries.  This is a common structure shared by all collections ranging 
from a space containing only empty data (null) to a space containing one atom (const (:)) to a space containing all mathematical objects (pass).  
The scope of these results means that pure data spaces can be thought of as a general point of view about mathematical objects, analogous to
sets with structure or Category Theory.  We highlight some of the theoretical differences before going on to examples.  
\begin{itemize}
\item{A set is defined by its contents, but a space is not.  Spaces (ap \{a\}) and (is a), for instance, contain the 
same data, are idempotent equivalent (ap \{a\})${\overset i\sim}$(is a), are isomorphic as spaces, and have 
the same semirings, but they are not the same space.  For example, ap \{a\}:(:) = a, but (is a):(:) = ().  Since every 
space contains its neutral data at least, it is not possible for a space to be empty.}
\item{Spaces have morphisms between them which compose associatively, but, unlike Category Theory, morphisms 
and spaces are ``made of the same substance.'' They are both merely data.  For instance, every data $A$ of pass is also 
a morphism pass$\cdot A\cdot$pass from pass to pass.  In contrast with categories and objects, a morphism 
$S{\overset {T\cdot S}\longrightarrow}T$ exists between any two spaces $S$ and $T$, so there are no closed categories 
where morphisms do not cross category boundaries.  Unlike Category Theory, a morphism can be morphisms between 
multiple spaces at once.  For instance, if $S$ and $T$ are commuting spaces, any fixed morphism $(S\cdot T){\overset F\rightarrow} U$ 
means that $F$ is a morphism from $S\cdot T$ to $U$, from $S$ to $U$ and is also a morphism from $T$ to $U$.}
\item{Category Theory is a theory of structures with corresponding morphisms.  In the theory of spaces, each space 
comes with both structure preserving ``homomorphisms'' and not-necessarily-structure-preserving morphisms, packaged together 
in a coherent semiring.} 
\item{In classical mathematics, one expects sets to have subsets, groups to have subgroups, rings to have ideals, and, in general, 
structures to have similar substructures.  Thus, it is not surprising that spaces can have subspaces, but unlike the classical situation, 
subspaces of a space can have completely different mathematical structures from the ambient space.  For example, we shall see that 
the space (is a b c d) has various free monoids, natural numbers, integers, the algebra of 4x4 natural number matrices and 
Gaussian Integers as different subspaces.  The extreme version of this is the space (pass) which contains {\it every} space as a subspace.  One can think of 
an algebraic space $S$ as having ``horizontal'' features coming from the subsemiring of homomorphisms and ``vertical'' features 
coming from general endomorphisms and subspaces.} 
\end{itemize}
Since this is a rich but unfamiliar viewpoint, the plan of the paper is to clarify the situation with examples, starting 
with foundational definitions which happen to be spaces (Figure 2) and proceeding ``organically''. 
The pure data view of spaces gives us a criteria for what is ``most organic'' and gives a way to systematically 
search for all spaces up to a specified data width and depth.  Although we mainly use the pure data perspective to 
choose spaces of interest, what we say about spaces will be put in semiring language whenever possible.  Aside from 
the organic and semiring themes, we will take special notice of algebraic data following the intuition that algebraic 
spaces are most likely to be mathematically interesting since they have, in a sense, earned a platonic meaning by 
transcending the foundational finite sequence structure of the system.  In semiring terms, we will thus focus on
the central homomorphisms of a space as the most likely features of interest. 

\begin{figure}[h]
\centering
\includegraphics[width=1.0\textwidth]{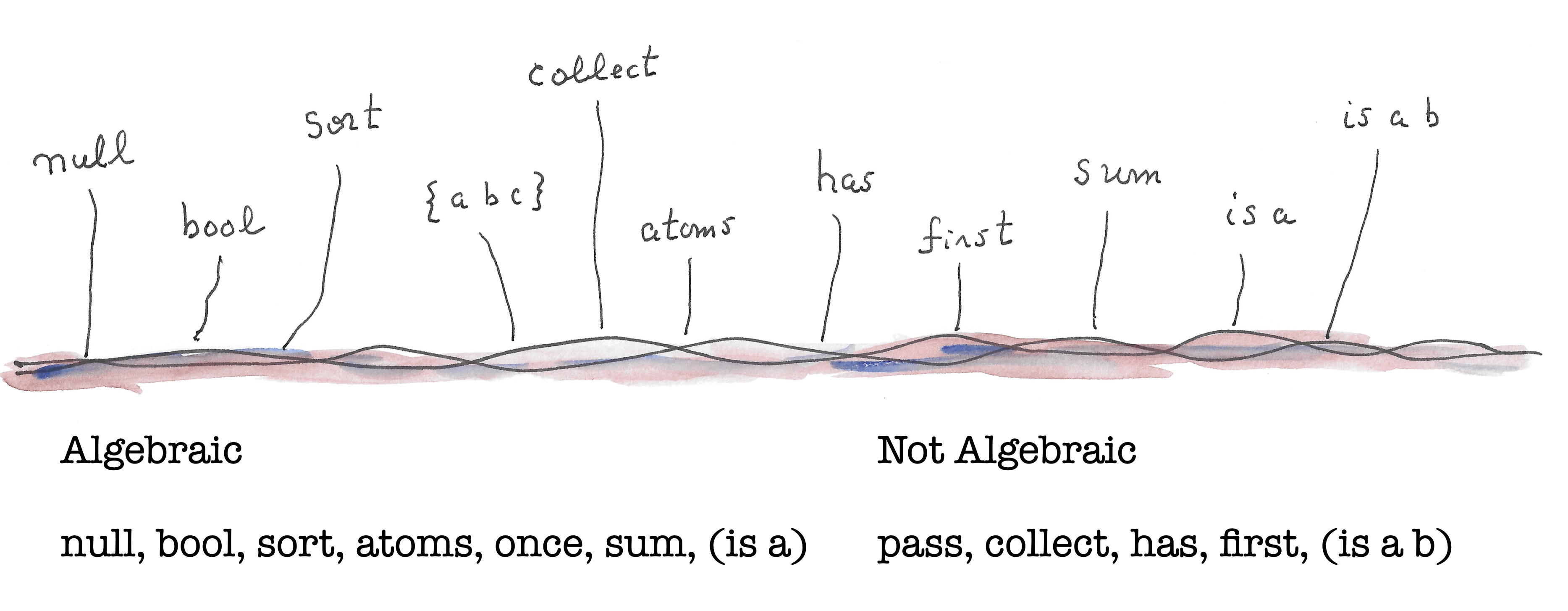}
\caption{{\it The simplest spaces emerging ``organically'' from the space (pass) of all pure data.  The algebraic spaces are commutative and, 
in a sense, transcend the foundational sequence concept, making them more ``platonic'' and more mathematically interesting than non-algebraic spaces.}} 
\end{figure}

\section{The Space of all data} 

     As a first example and for orientation, let's consider the space of all pure data.  Since (pass:X)=X for all data X, the distributive space (pass) 
contains all data.  Since spaces are fixed points on their contents, pass is the unique space containing all data.  Every data $X$ is both in pass
and is also an endomorphism of pass, since $X$ is equal to pass$\cdot X\cdot$pass.  Thus, the semiring of pass is also the collection of all data, with 
addition ($X\oplus Y$)=($X+Y$), and multiplication $(X\cdot Y)$.  The units of the semiring are $1$=pass$\cdot$pass$\cdot$pass=pass, 
and $0$=pass$\cdot$null$\cdot$pass=null respectively.   

Consider the properties defined in section 6.  
\begin{itemize}
\item{{\bf Subspaces.}  Since every data is an endomorphism of pass, every space is a subspace of pass.  The constants of pass are the data $K$ such that 
$K\cdot 0$=$K$, in other words, the constant spaces (const $K$) are the constants.}
\item{{\bf Homomorphisms.} Since pass is distributive, the homomorphisms of pass are the distributive data in pass, including spaces, such as (ap \{a\}) 
and non-spaces such as (ap \{B B\}).} 
\item{{\bf Group of units.} The units of pass are permutations.  The central constants of pass are the data which 
are invariant under permutations.  These are the sequences of identical atoms such as (:) (:) (:) or (a a a a a).  These are interpreted as ``organic natural numbers.''}
\item{{\bf Central endomorphisms} The maximally platonic, maximally significant content of pass are the central homomorphisms, which include 
subspaces such as (ap const a), (atoms), (is b).  These are isomorphic spaces, each representing natural numbers with natural number 
addition.  This, and the direct connection of bool with the underlying pure data context means that 
both the natural numbers and bool are natural starting points for exploration.} 
\item{{\bf Semialgebra} In pass, there is a bijection between the homomorphisms of pass and the constants of pass given by $X\leftrightarrow ({\rm ap}\ X)$.  Thus,
pass is a semialgebra as well as a semiring with $X\star (Y+Z)$ equal to $(X\star Y)+(X\star Z)$ where $X\star Y$ is defined to be (ap $X$)$\cdot Y$.}
\item{{\bf Idempotents} The idempotent equivalence classes of pass includes all idempotent data where (pass) is the unique maximum idempotent and where all constants are in the minimal equivalence class.}
\item{{\bf Semilattice} Semilattice spaces include (bool), (once) and (const $K$) for any data $K$.  The space bool is special in several ways.  
It is a semilattice, a field, a central endomorphism, and is isomorphic to any space with exactly two data.}
\end{itemize}

\begin{figure}[h]
\centering
\includegraphics[width=0.8\textwidth]{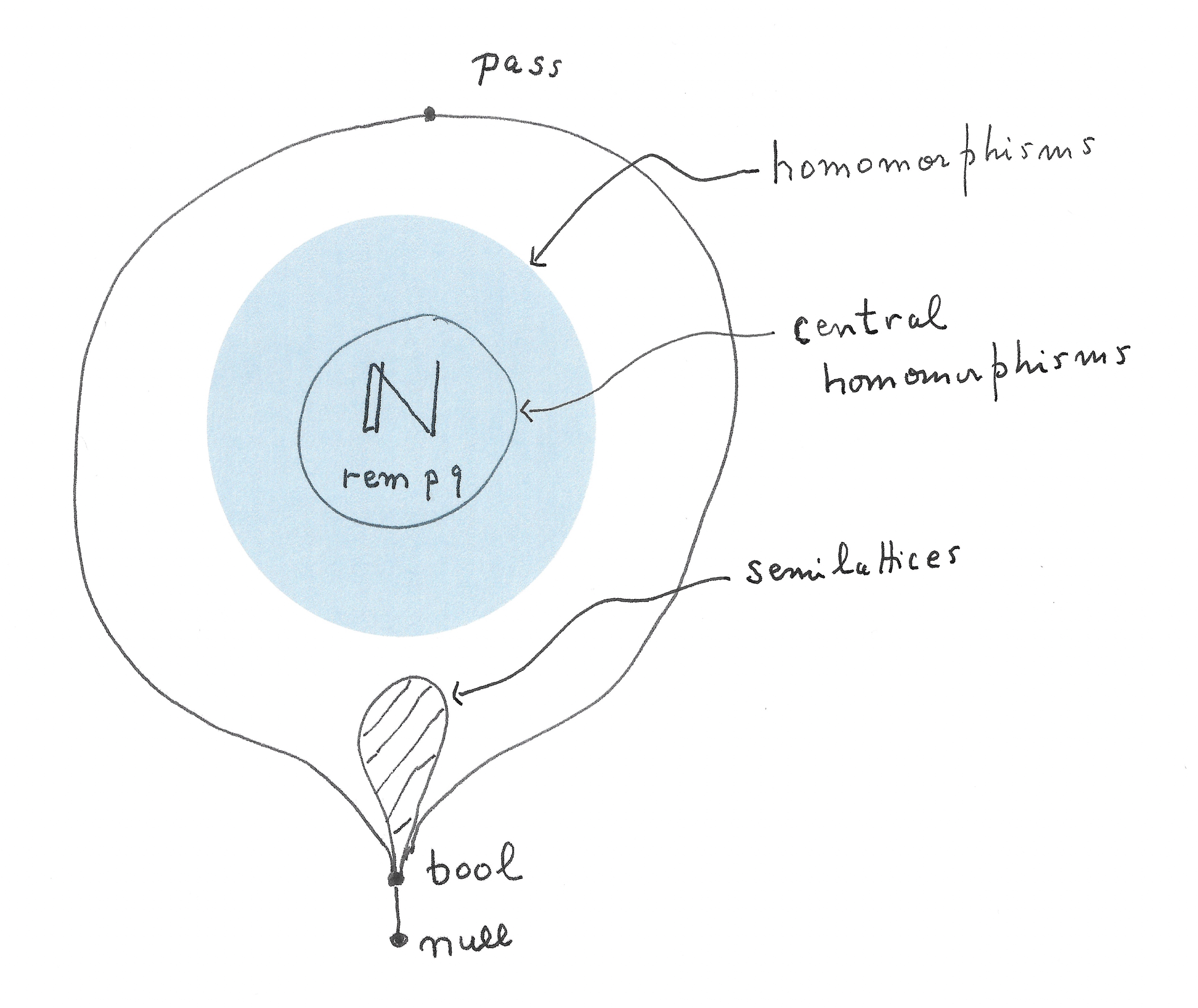}
\caption{{\it As a collection, pass contains all pure data; that is, all objects with mathematical 
meaning as well as all computations.  As a space, it has a specific structure which is typical of other examples in the text.  
In the idempotent order, pass is the unique maximum and null is minimal as are all constants.  
Each data in pass is also an endomorphism of pass.  The homomorphisms of pass are the distributive data.  
The central homomorphisms of pass are a single isomorphism class equivalent to the ``organic'' natural 
numbers.  Spaces that are fields includes bool, null, and the rem(p,p) subspaces of ${\mathbf N}$, equal to ${\mathbf F}_p$ for prime $p$. 
The space bool is the unique space with two data and is a semilattice as well as a field.}}
\end{figure}

In classical mathematics, the concept of a Set intentionally gives no information about its contents.  Thus, the set of all sets is vast, but 
not interesting.  In our case, (pass) is vast, but also has particular features which will be common to all of the spaces considered here. 
Figure 3 gives a visual overview of pass, highlighting the features of interest and points for exploration in the following sections. 

\section{Null}

    Null is the distributive algebraic space defined by (null:$X$)=() for all $X$.  Since null$\cdot X\cdot$null=null, null is the only endomorphism 
of null, and so 1=0 in the semiring, as is true for all other constant spaces.  Since $0\cdot(0\oplus 0)=(0\cdot 0) \oplus (0\cdot 0)=0$, the one 
endomorphism is a homomorphism as well as a subspace.   Null is minimal in the idempotent order. 

\section{Organic Numbers}

      We have seen that ``organic natural numbers'' appear as sequences of identical atoms where we identify (:) (:) (:) or (a a a) with ``3,'' just depending 
 on a conventional choice of atom.  The spaces where these numbers are data are the maximally platonic central homomorphisms of 
 pass.  For instance, the spaces (atoms), (ap const (:)), (ap const a), (ap \{\#\}), (is a) all are central homomorphisms and each represents natural numbers 
 with natural number addition.  Since these spaces are isomorphic and, therefore, have the same semirings, we can choose any one for investigation.  
 For convenience we will choose (is a) as the particular space to start with.  For convenience, denote $n$ concatenations of data $A$ by $A^n$ so so that 
 ``5'' represented by (a a a a a) and can be written as a$^5$.  
 The spaces (is a), (is a b), (is a b c),$\dots$ are {\it organic numbers}.  
      
\subsection{is a}     
     
     Start with $N$=(is a), the {\it organic natural numbers}.  
The first thing to notice about $N$ is that $N$ contains one data a$^n$ for each 
natural number $n$, and addition in $N$ is natural number addition, 
since ($N$:$X$\ $Y$) is the natural number sum of ($N$:$X$) and ($N$:$Y$).  Since homomorphisms are distributive, a homomorphism $h$ of $N$ is determined by ($h$:a), so that $h$ is natural number multiplication by some natural number, 
for instance, ($h$ = ap const a\ a\ a) is multiplication by $3$.  There is one such homomorphism for each data in $N$, and so $N$ is a semialgebra and 
is identical to the standard semiring of natural numbers with multiplication distributing over addition.
     
    Subspaces of $N$ are indicated by subspace endomorphisms.  As always, the subspace $1$ is the whole of $N$ and the subspace $0$ is the neutral data of $N$, which, 
since $N$ is distributive, is the empty sequence.  It is easy to guess some subspaces of $N$.
\begin{itemize} 
\item [1.] {\bf Saturation subspaces} An endomorphism (min a a) is a subspace which effectively computes the minimum of 2 and its argument.  Thus (min a a) contains the data (), (a), and (a a). 
\item [2.] {\bf Modular arithmetic subspaces} An endomorphism which removes three (a) atoms while possible, such as (while remove a a a) is a subspace which computes the sum modulo 3.  Note 
that (while remove a a a) contains the same data as (min a a), but have a different additive structures.  
\end{itemize}
These unsurprising subspaces have a slightly more surprising generalization.  For $p,q\ge 1$, let rem($p$,$q$) be the endomorphism defined by 
\begin{equation}
{\rm while}\ n\ge p,q: {\rm remove\ p\ from\ n}
\end{equation}
which results in rem($p$,$q$)($n$) = min($n$\ mod\ $p$,$q-1$), specializing to addition modulo p for rem(p,p).  Since a subspace of $N$ is exactly an idempotent 
endomorphism of the semigroup (${\mathbb N},+$), it is known from the theory of semigroups that this list is complete\cite{semigroup}.  The only subspaces of $N$ are $1$, constants and rem($p$,$q$) for $p,q\ge 1$, and the only field subspaces of $N$ are rem($p$,$p$) for prime $p$. 

\begin{figure}[h]
\centering
\includegraphics[width=0.7\textwidth]{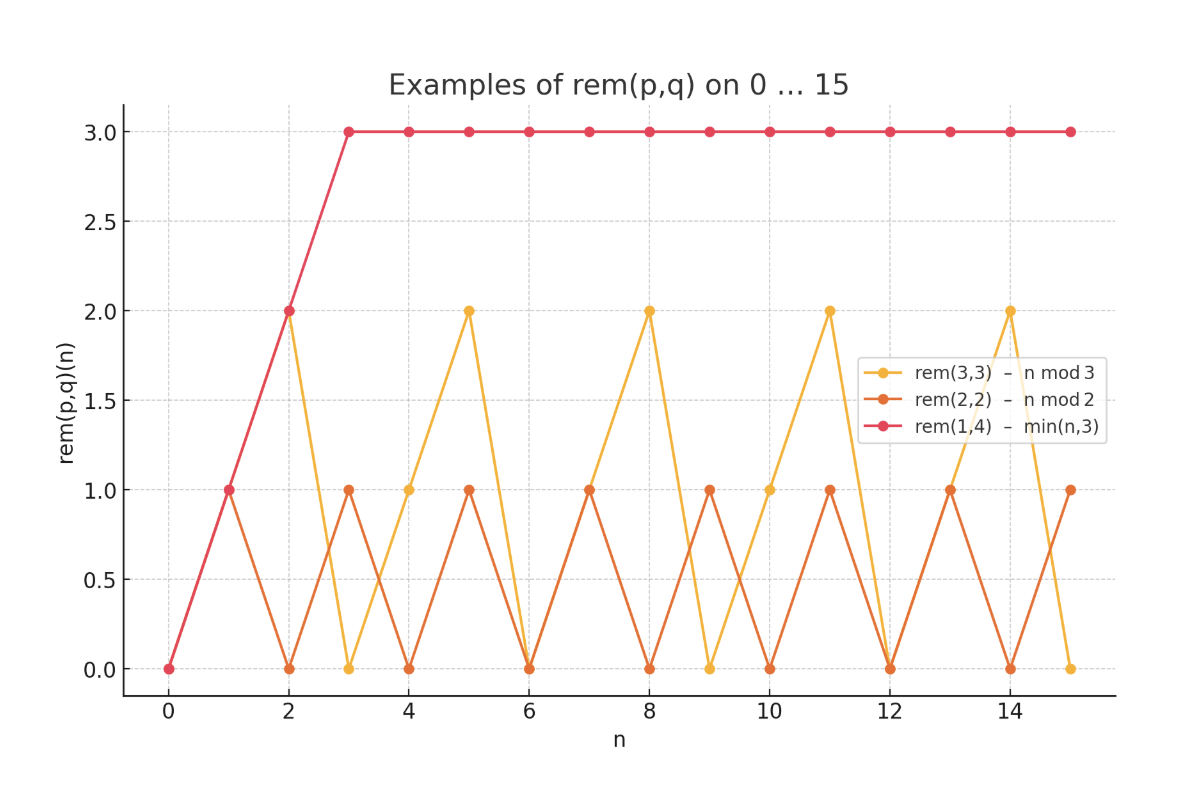}
\caption{{\it The subspaces of the organic natural numbers are $N$ itself, the constants subspaces, and the rem(p,q) subspaces, 
which specialize to natural numbers modulo $p$ and ceiling functions less than or equal to $q$.}} 
\end{figure}

\subsection{is a b} 

The second simplest organic number is (is a b), which we can refer to as N$_2$ for brevity.  Since N$_2$ is distributive, the homomorphisms of N$_2$ 
are the distributive endomorphisms.  Consider subspaces N$_2$.  

\begin{itemize}
\item{The subspaces 1 and 0 give the whole of N$_2$ and just the empty sequence as subspaces, respectively.  
Projections N$_2$$\cdot$(is a)$\cdot$N$_2$ and N$_2$$\cdot$(is b)$\cdot$N$_2$ produce two $N$-isomorphic subspaces. The absolute value 
subspace (ap const a) is also isomorphic to $N$.}

\item{Let `reduce' be the subspace of N$_2$ which removes (a b) and (b a) subsequences until saturation.  Thus, every data in reduce
is either a$^n$ or b$^n$ for some natural number $n$.  A homomorphism $h$ satisfies $h:X\ Y$ = reduce $:(h:X)\ (h:Y)$, so 
$h$ is determined by ($h$:a) and ($h$:b).  If $h$ is a central homomorphism, ($h$:b) is ($h$:a) followed by an 
a$\leftrightarrow$b swap.  Thus, if we interpret a$^n$ as the integer $n$ and b$^n$ as the integer $-n$, then applying $h$ 
is standard distributive integer multiplication by the integer ($h$:a).  Since there is one central homomorphism for each data in the subspace reduce, 
reduce is a central semialgebra isomorphic to the standard ring $\mathbf Z$ of integers.}

\item{If {\it sort} is the endomorphism which does lexical sorting of N$_2$ data, then sort is also a subspace of N$_2$ where the 
data of sort can be written $a^m b^n$ for $m,n\ge 0$.  As in the previous case, a homomorphism M of sort is defined by  
\begin{itemize}
\item [] $M: a\mapsto a^{m_{11}} b^{m_{21}}$, and 
\item [] $M: b\mapsto a^{m_{12}} b^{m_{22}}$ 
\end{itemize}
for some choice of 
$
\left (
\begin{array}{cc} 
m_{11} & m_{21} \\ m_{12} & m_{22}  
\end{array}
\right ) 
$
so the action of M is standard matrix multiplication.  Thus, the data of sort are pairs of natural numbers and the 
homomorphisms are 2x2 matrices with natural number entries.  The central homomorphisms of sort 
commute with the involution 
$
\left (
\begin{array}{cc} 
0 & 1 \\ 1 & 0 
\end{array}
\right ) 
$
and so the central homomorphisms are matrices 
$
\left (
\begin{array}{cc} 
m & n \\ n & m 
\end{array}
\right ) 
$
for natural numbers $m$ and $n$.  Thus, there is exactly one central homomorphism for each data in sort, 
and sort is, therefore, a central semialgebra, equivalent to the standard matrix semiring ${\rm Mat}_{2\times2}(\mathbf N)$.
}

\item{Lexical sort followed by reducing (b b) to (b) results in a subspace equivalent to $\mathbf N\times\{0,1\}$ with addition
defined by $(n,\alpha)+(m,\beta)$=$(n+m,\alpha\vee\beta)$.}

\item{Lexical sort followed by reducing $a^n b^m$ to $a^{(n/g)} b^{(m/g)}$ where $g$=GCD$(n,m)$ is the greatest common divisor.  
This is the space of positive rational numbers, but with non-standard `mediant' addition where $(n/m)+(n'/m')$ is $(n+n')/(m+m')$\cite{mediant}.}

\end{itemize}

\subsection{is a b c d...} 

     It is clear that in going to higher organic numbers, some of the features of $N$ and $N_2$ will repeat for 
 $N_n$=(is a b c...-{\it n atoms}-...).  $N_n$ sorted 
 will always result in a space with homomorphisms ${\rm Mat}_{n\times n}(\mathbf N)$.  Sorting followed by cancelling 
 atoms in pairs will result in ${\rm Mat}_{(n/2)\times (n/2)}(\mathbf Z)$.  For example, in (is a b c d) with subspace 
 (a b)=(b a)=(c d)=(d c)=(), homomorphisms are ${\rm Mat}_{2\times 2}(\mathbf Z)$.  Central homomorphisms of $N_4$ sorted and reduced 
 must commute with 
 $
J=
\left (
\begin{array}{cc} 
0 & -1 \\ 1 & 0 
\end{array}
\right ) 
$
and are thus matrices of the form 
$
\left (
\begin{array}{cc} 
x & -y \\ y & x 
\end{array}
\right ) 
$
and we have the central semialgebra equivalent of the Gaussian Integers ${\mathbf Z}[i]$.  
Figure 5 shows similar results for some of the more notable subspaces.  
Notice that this is a success from the organic gardening perspective.  Important mathematical 
structures are emerging quite naturally.  Even though both semiring operations 
are associative, a non-associative algebra such as the integer Octonions appears in the 
same way as the other spaces with homomorphisms implementing the non-associative octonion 
multiplication\cite{octonions}.  
 
\begin{figure}[h]
\centering
\includegraphics[width=1.1\textwidth]{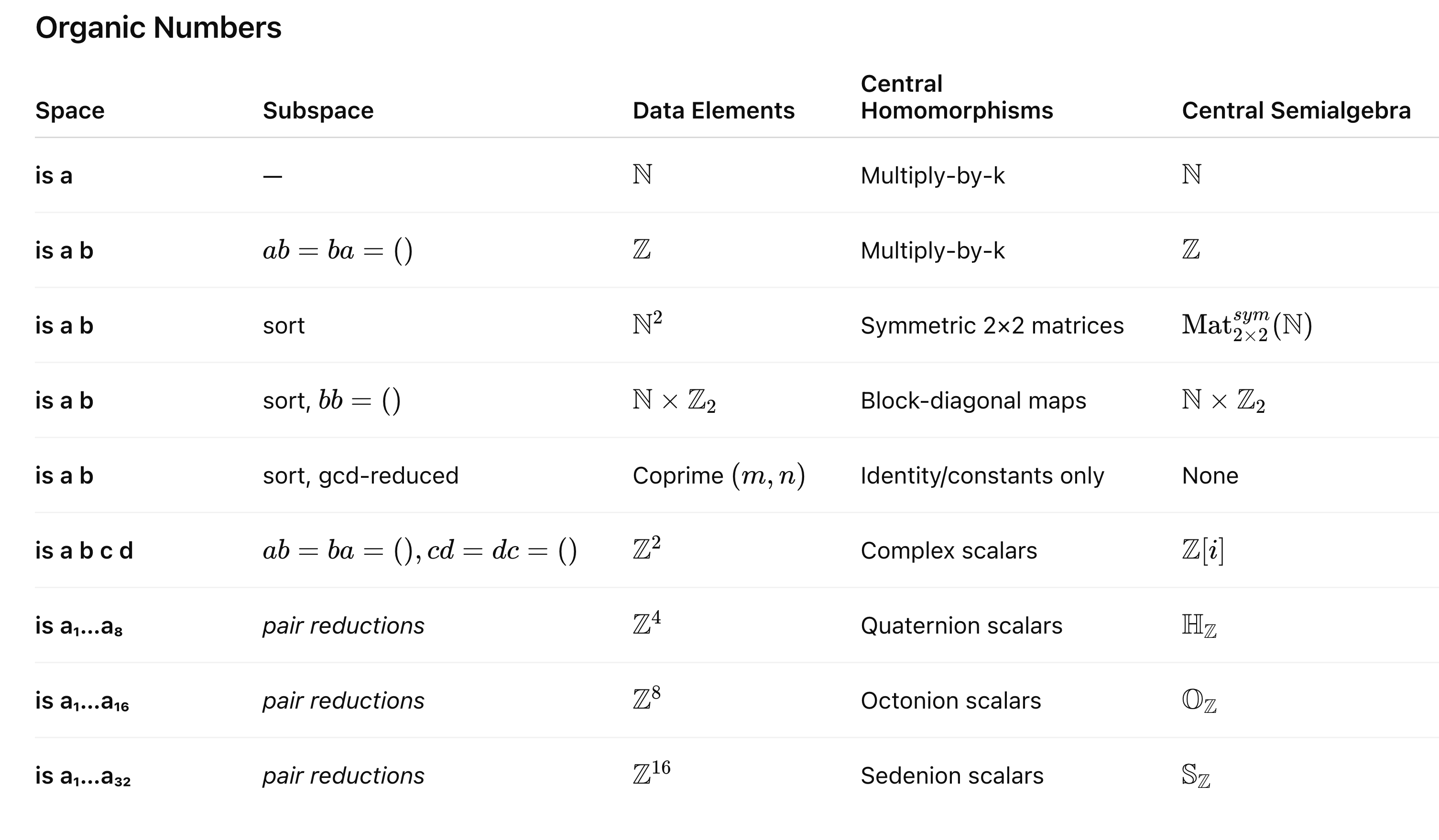}
\caption{{\it Classical algebras appear as organic numbers.}}
\end{figure}
\noindent Although many interesting spaces appear as organic numbers, some simple spaces appear to be missing.  It's noticeable, 
for instance, that standard rationals do not appear.  We would like to understand why.  

\section{Rationals} 

       In the naive version of the organic theme, we would define a context with the basic combinatoric definitions and then 
search the space of pure data, up to, say, some specified width and depth.  Then, for each data, we would test if it is 
a space or not.  Ideally, all of the most important mathematical structures would be found this way.  They would appear ``organically.'' 
This has actually worked fairly well up through the organic numbers (Figure 5), but there we noticed that only a mediant-addition 
version of the rational numbers appeared as an organic integer.  The theory of fields makes the reason for this clear.  
Any non-constant subspace of an organic number space (is a b c...) will have proper non-constant subspaces formed by 
projecting onto a subset of (a b c...).  This means that the only non-trivial fields in the organic numbers are ${\mathbf F}_p$, 
the subspaces rem($p$,$p$) for prime $p$ of (is a).  As a result, we take a more active design approach to creating a space that represents the standard rationals, while still aiming to be as organic as possible in the sense of only using basic combinatoric definitions.  For illustrative purposes, we follow a broadly applicable plan for defining a space corresponding to some algebra of interest.  The plan is as follows. 
 
 \begin{itemize}
 \item[(a)]{Define suitable atoms to represent the objects of interest.}
 \item[(b)]{Create a base space; typically a distributive sequence of atoms defined in (a).} 
 \item[(c)]{Invent a sequence of mutually commuting subspaces of the base space, until the sum in the product of the subspaces has 
 the desired sum in a mathematical sense.}
 \item[(d)]{The homomorphisms of the resulting product subspace then define a module-like product.  If there is a suitable mapping 
 from constants to homomorphisms, one has a semialgebra or a central semialgebra.}
 \end{itemize}
 In the case of rational natural numbers, we proceed as follows.
\begin{itemize} 
\item[(a)] { 
Represent single rational numbers with $q$-atoms such as ($q$ a a: a a a) representing `2/3.'  
It is convenient and illustrative to define a ``fancy'' $q$-atom that is not a direct fixed point, but is only equal to a fixed point. Define $q$ by 
\begin{itemize}
\item{($q$ $A$ $B$ : $B$) $\mapsto$ ($q$ $A$:$B$)}
\item{($q$ $A$ : $A$ $B$) $\mapsto$ ($q$ $A$:$B$)}
\end{itemize}
so that $A$ and $B$ components of each $q$-atom are coprime natural numbers. 
}
 \item[(b)] {For the base space, we can simply choose the distributive space of $q$-atoms, modified only by 
 filtering to guarantee that ``numerator'' and ``denominator'' are sequences of a-atoms, and excluding $q$-atoms with empty denominators. 
}
 \item[(c)] {Define mutually commuting subspaces of $Q$ approaching the desired properties of rationals. 
\begin{itemize}
\item{ADD is the subspace defined by $(q\ X_1:Y_1)\ (q\ X_2:Y_2)$ maps to 
$(q\ ({\rm pr}\ X_1:Y_2)\ ({\rm pr}\ X_2:Y_1) : ({\rm pr}\ Y_1:Y_2))$, where 
`pr' is (ap const a)$\cdot$ar, so that (pr $A$:$B$) is the natural number product of $A$ and $B$.}
\item{NORM is the subspace of $Q$ defined by ($q$\ $A$:$A$) $\mapsto$ (q\ a:a).}
\item{NORM0 is the subspace of $Q$ defined by ($q$\ ():$b\ B$) $\mapsto$ ().}  
\end{itemize}
The subspace S=(ADD$\cdot$NORM$\cdot$NORM0) thus contains the empty sequence (representing the zero rational) and a single q-atom (representing non-zero rationals).  The sum in $S$ is standard rational number addition in $\mathbb Q^+$.

 }
 \item[(d)] {If we define (product ($q\ X_1:Y_1$) : ($q\ X_2:Y_2$)) to be the standard rational product $(q\ ({\rm pr}\ X_1:X_2): ({\rm pr}\ Y_1:Y_2))$, then 
(ap product $R$) is a homomorphism of $S$ and is the standard distributive rational multiplication by rational number $R$, making  
$S$ is a central semialgebra and a field.}
 \end{itemize} 
 We find that the standard rational natural numbers do appear to have a representation in this framework which is natural in the sense that 
 it is a simple construction from basic combinatoric definitions.  It is, however, created via a bespoke intentional process steering towards 
 an intended result, rather than as a space appearing with essentially no effort. 
 
\section{Number Sequences}

      Given any space $S$, a common functor-like construction produces a corresponding distributive space of 
$S$-values stored in some chosen $s$-atom.  An example with the organic natural numbers will illustrate how this works. 
Given the organic naturals $N$=(is a), from Section 9.1, we can go from a space containing single natural numbers to a space 
containing sequences of natural numbers.  
Let $\mathbb{N}$ be 
\begin{equation}
{{\rm ap}\ ({\rm put}\ n)\cdot N \cdot ({\rm get}\ n)}
\end{equation}
where $n$ is a pre-defined atom.  Since $({\rm put}\ n)\cdot N \cdot ({\rm get}\ n)$ is idempotent, ${\mathbb N}$ is a distributive space containing 
sequences of $n$-atoms containing $N$-data, such as 
\begin{equation} 
T = (n:a\ a\ a)\ (n:)\ (n:a\ a)\ (n:a)\ (n:).
\end{equation} 
Unlike the case of $N$, the action of ${\mathbb N}$ is merely to concatenate sequences.  It's easy, however, to identify 
natural number sum again as one of several subspaces 
\begin{itemize}
\item sum:$T$ = ($n$:a a a a a a)
\item sort:$T$ = ($n$:) ($n$:) ($n$:a) ($n$:a a) ($n$:a a a)
\item min:$T$ = ($n$:)
\item first:$T$ = ($n$:a a a)
\end{itemize} 
All but the last are algebraic subspaces and, predictably, have recognizable names.  
Some endomorphisms of ${\mathbb N}$ are inherited from $N$ as follows.  Define inner:$F$ to be 
\begin{equation}
{\mathbb N}\cdot({\rm put}\ n)\cdot F\cdot({\rm get}\ n)\cdot{\mathbb N}
\end{equation}
so that (inner:$F$):$X$ means letting $F$ act on the concatenated $N$-contents of $X$, returning the result in a single $n$-atom.  The sum endomorphism 
above, for instance, is equal to inner:$N$.  The construction of Equation 11 works for any space, and so we can define a ``functor'' {\it Seq} to be 
\begin{equation}
{\rm ap}\ \{({\rm put}\ {\rm A})\cdot {\rm B} \cdot ({\rm get}\ {\rm A})\}.
\end{equation}
Then {$\mathbb N$} is equal to (Seq $n$:$N$) and given any atom $s$ and any space $S$, (Seq $s$:$S$) is the space of $S$-values stored in $s$-atoms.  
Similarly, if we define {\it inner} to be  
\begin{equation}
S\cdot\{({\rm put}\ s)\cdot{\rm B}\cdot({\rm get}\ s)\}\cdot S
\end{equation}
then inner:$f$ is the inner version of an endomorphism $f$ of $S$.  Many general constructions of this type are clearly possible.  For example, 
since the endomorphism last2 = (inner:last 2) sums the last two atoms of a data of ${\mathbb N}$, (series last2) generates the fibonacci sequence. 

\section{Sets}

    We have seen that familiar mathematical structures organically appear as pure data spaces rather than as sets with additional algebraic structure.  
It is natural to ask whether ordinary sets also appear organically.  

     Consider $S$=first$\cdot$(is a b c).  $S$ is a distributive space containing \{(),a,b,c\} as data.  Endomorphisms of $S$ are  
$S\cdot X\cdot S$ for some data $X$ and, thus, can be any ordinary function from \{(),a,b,c\} to \{(),a,b,c\}.  Homomorphisms of $S$ 
are endomorphisms which map () to ().  Thus, $S$ is, roughly speaking, the set \{a,b,c\} but with awkward complications due 
to the forced presence of ().  Thematically, we attribute this to the ``flaw'' that even though $S$ has the key property 
that $f\oplus f$=$f$ for all endomorphisms, $S$ is not algebraic and, thus, we expect it not to be perfectly mathematical.  

\begin{figure}[h]
\centering
\includegraphics[width=0.5\textwidth]{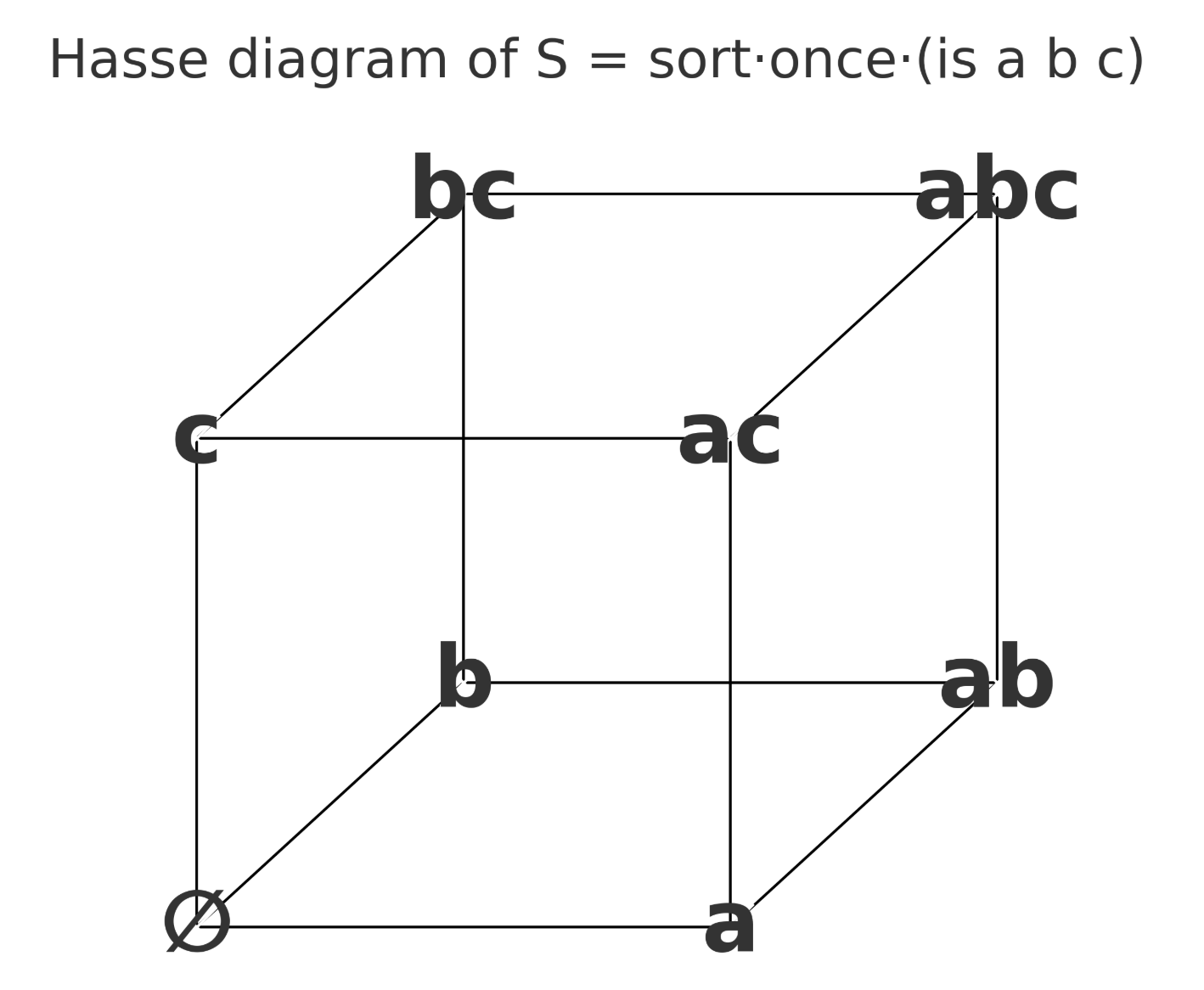}
\caption{{\it Hasse diagram of the constant endomorphisms of $S$=sort$\cdot$once$\cdot$(is a b c).  The data of $S$ are subsets 
of \{1,2,3\} ordered by inclusion.}}
\end{figure}

     The partial success of first$\cdot$(is a b c) suggests considering an {\it algebraic} space where $f\oplus f=f$.  
These are the {\it semilattice} spaces.  The space $S$=sort$\cdot$once$\cdot$(is a b c) is an example.  
Instead of containing the elements of \{a,b,c\}, $S$ now contains the $2^3$ subsets \{\{\}, \{a\}, \{b\}, \{c\}, \{a b\}, \{a c\}, \{b c\}, \{a b c\}\} 
represented as sorted non-repeating sequences.  The endomorphisms of $S$ are general endofunctions of ${\mathcal P}(\{a,b,c\})$.  The semilattice of 
constants of $S$ is the standard semilattice of subsets of \{a,b,c\} ordered by the set union operation.  
The unit endomorphisms are the permutation group $S_3$ permuting a, b and c.  The central endomorphisms are $S_3$-equivariant functions, and the homomorphisms 
are morphisms of the standard set union, as shown in Figure 6.  In other words, $S$ has everything that one mathematically expects, without extraneous features or 
unexpected constraints.  Thus, the intuition that algebraic spaces will be more platonic and mathematically perfect is encouraged.  

\section{Bool} 

      As mentioned already, the space (bool) is likely to be mathematically interesting, both because of its direct connection to 
the underlying pure data context and, because it has a number of special properties.  Although bool is not a homomorphism of pass,
it is algebraic, a central endomorphism, a semilattice, a field, and bool is the only space containing exactly two data. 
Because bool only contains:  () for {\it true} and (:) for {\it false}, there are only four endomorphisms:  
the identity ID=bool, two constants: TRUE=const (), and FALSE=const (:), and one involution, NOT=bool$\cdot$not$\cdot$bool.  
\begin{table}
\begin{tabular}{| l | l | l | l | l |  }
$f\cdot g$ & ID & TRUE & FALSE & NOT  \\
\hline
ID &  ID & TRUE & FALSE &  NOT \\
TRUE & TRUE & TRUE  & TRUE & TRUE \\
FALSE & FALSE  & FALSE & FALSE & FALSE   \\
NOT & NOT & FALSE & TRUE & ID \\
\hline
\end{tabular}
\begin{tabular}{| l | l | l | l | l |  }
$f\oplus g$ & ID & TRUE & FALSE & NOT  \\
\hline
ID &  ID & ID & FALSE & FALSE \\
TRUE & ID & TRUE  & FALSE & NOT \\
FALSE & FALSE  & FALSE & FALSE & FALSE   \\
NOT & FALSE & NOT & FALSE & NOT \\
\hline
\end{tabular}
\caption{{\it Product and sum of the four morphisms ID, TRUE, FALSE, NOT of the space bool.  ID is the unit of multiplication and TRUE is the 
unit of addition, and we have $f\oplus g=g\oplus f$ and $f\oplus f=f$ since bool is a semilattice space.}}
\end{table}
\begin{table}
\centering 
\begin{tabular}{r c c c c c c c r r l}
\hline\hline
$e_i \in {\mathbb L}_2$ & 0 & T & F & TT & TF & FT & FF & standard & generalized & description \\ [0.5ex] 
\hline
$e_1$  & T & T & T & T & T & T & T & TRUE & {\bf always} & always true \\
$e_2$  & T & T & T & T & T & T & F & OR & {\bf any} & any are true \\
$e_3$  & T & T & F & T & F & F & T & XNOR & {\bf even} & even (:)s \\
$e_4$ & T & T & F & T & F & F & F & AND & {\bf all} & all are true \\
$e_5$ & F & F & T & F & T & T & T & NAND & {\bf notall} & not all are true \\
$e_6$ & F & F & T & F & T & T & F & XOR & {\bf odd} & odd (:)s \\
$e_7$ & F & F & F & F & F & F & T & NOR & {\bf none} & none are true  \\
$e_8$ & F & F & F & F & F & F & F & FALSE & {\bf never} & never true \\
\hline
Number of (:)   & 0 & 0 & 1 & 0 & 1 & 1 & 2 &  \\ 
\hline
\end{tabular}
\caption{{\it The eight inner morphisms of ${\mathbb L}_2$ slightly generalize the eight standard symmetric binary boolean operators.
In this slightly generalized context, one can see that the standard names for these operators are not perfectly descriptive.}}
\label{table:L2}
\end{table} 
Table 1 shows the bool semiring.  The only proper subspaces of bool are the constants TRUE and FALSE, so bool is a field. The group of units is ID and 
the involution NOT, so ID is the only central endomorphism.    

\subsection{Boolean Sequences}

Repeating the sequencing operation from Section 12, we can let $\mathbb L$=(Seq $b$:bool), 
be the distributive space of bool-valued data stored in pre-defined $b$-atoms, so a typical data in $\mathbb L$ is 
\begin{equation}
(b:)\ (b:)\ (b:(:))\ (b:)\ (b:(:))\ (b:(:))
\end{equation}
We write such sequences replacing ($b$:) with T, ($b$:(:)) with F and the empty sequence with 0, so that the above is written TTFTFF.   
Consider the shortest subspaces of $\mathbb L$ first.  
The subspace (${\mathbb L}\cdot {\rm first}\cdot{\mathbb L}$) contains the $\mathbb L$ sequences of length less than or equal to 1 and similarly,  
(${\mathbb L}\cdot {\rm first\ 2}\cdot{\mathbb L}$) is the subspace of sequences with length less than or equal to 2.  Lets refer to these spaces as 
${\mathbb L}_1$ and ${\mathbb L}_2$ respectively.  
${\mathbb L}_1$ contains data \{0,T,F\} and ${\mathbb L}_2$ contains data \{0,T,F,TT,TF,FT,FF\}. 
For ${\mathbb L}_1$, we can specify an endomorphism $f$ of ${\mathbb L}_1$ by listing the values of $f$ on 0,T,F in standard order.  So, for example, 
0TF denotes the identity endomorphism of ${\mathbb L}_1$.  
Since ${\mathbb L}_1$ has 27 morphisms, we can be completely explicit about this space.  
Figure 7 lists the 27 endomorphisms and indicates endomorphisms with special properties. 
Even with relatively small space like ${\mathbb L}_1$, all of 
the classes of endomorphisms appear in a nontrivial way.

\begin{figure}[h]
\centering
\includegraphics[width=0.8\textwidth]{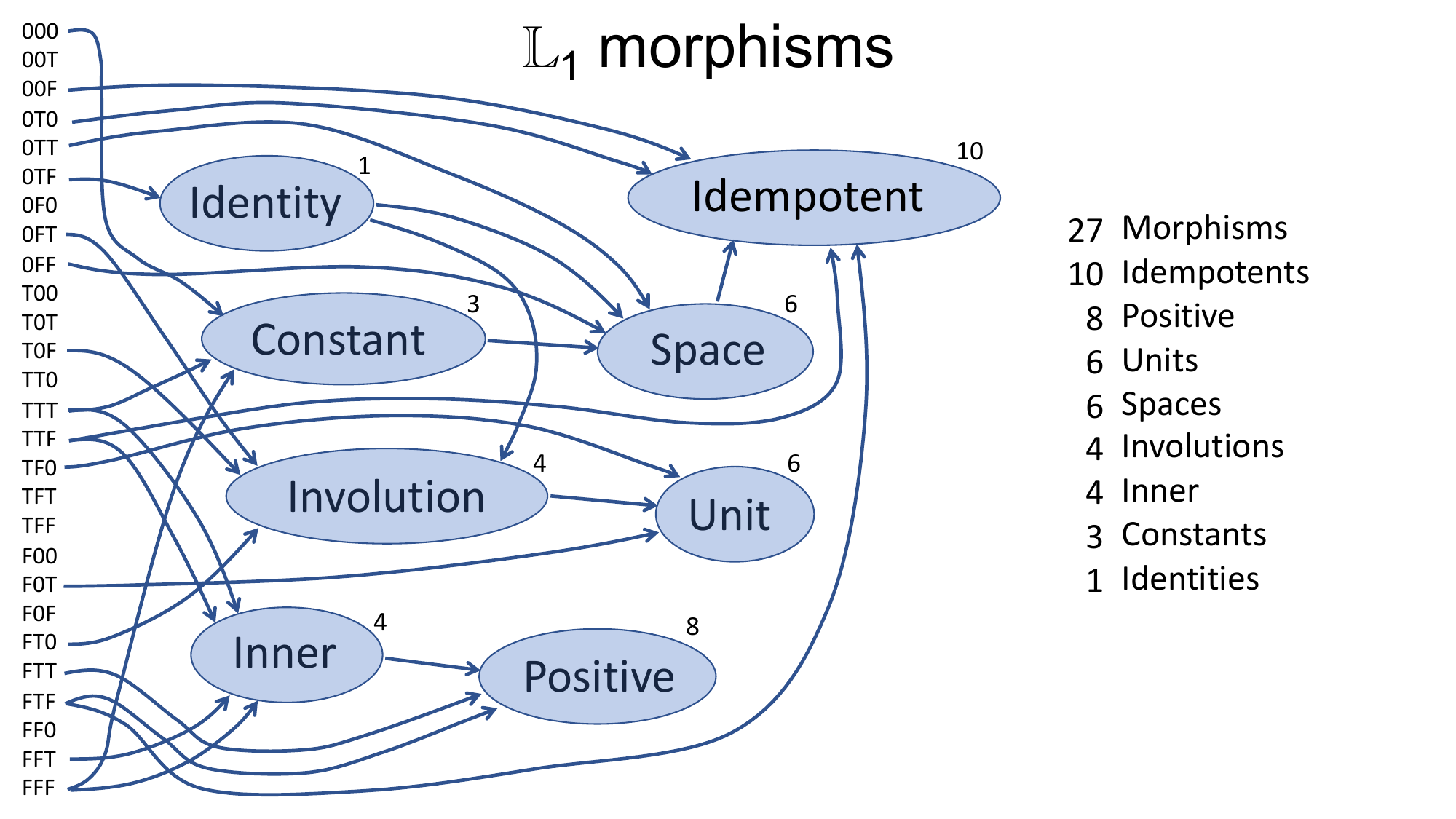}
\caption{{\it The endomorphisms of ${\mathbb L}_1$ are maps from \{0,T,F\} to itself. These are specified as the three function 
values in the column on the left.  The `positive' endomorphisms are the subsemiring of endomorphisms which are never neutral.  Even with a small space of this size, the $3^3=27$ all of the classes of endomorphism identified in the text 
appear in a non-trivial way.}}
\end{figure}

    Moving to ${\mathbb L}_2$, there are already $7^7=823,543$ endomorphisms; far too many to be as explicit as in the case of ${\mathbb L}_1$.  
 We can, however, at least examine the eight inner endomorphisms 
 \begin{equation}
 {\mathbb L}_2\cdot ({\rm put}\ b)\cdot X \cdot ({\rm get}\ b) \cdot {\mathbb L}_2
 \end{equation}
These depend only on the total number of (:) values in the $b$-atoms of input, and each 
inner endomorphism returns exactly one $b$-atom.  
The inner endomorphisms of ${\mathbb L}_2$ are central homomorphisms, so we have reason to expect them to be
mathematically interesting.   
Indeed, Table 2 shows the eight endomorphisms and their values on \{0,T,F,TT,TF,FT,FF\}.  Looking at the values for the last four entries \{TT,TF,FT,FF\} 
we recognize that these are slight generalizations of the eight standard symmetric binary boolean operators.  The values of the eight morphisms 
on \{0,T,F\} suggests that the standard binary operator names are not perfectly descriptive in this slightly larger context.  

    Further analysis of ${\mathbb L}$ is beyond the scope of this paper, but we note that the obvious isomorphism between (is a b) and ${\mathbb L}$ 
means that all of the semiring structures and subspaces found in (is a b) have copies in ${\mathbb L}$.  

\begin{figure}[h]
\centering
\includegraphics[width=1.0\textwidth]{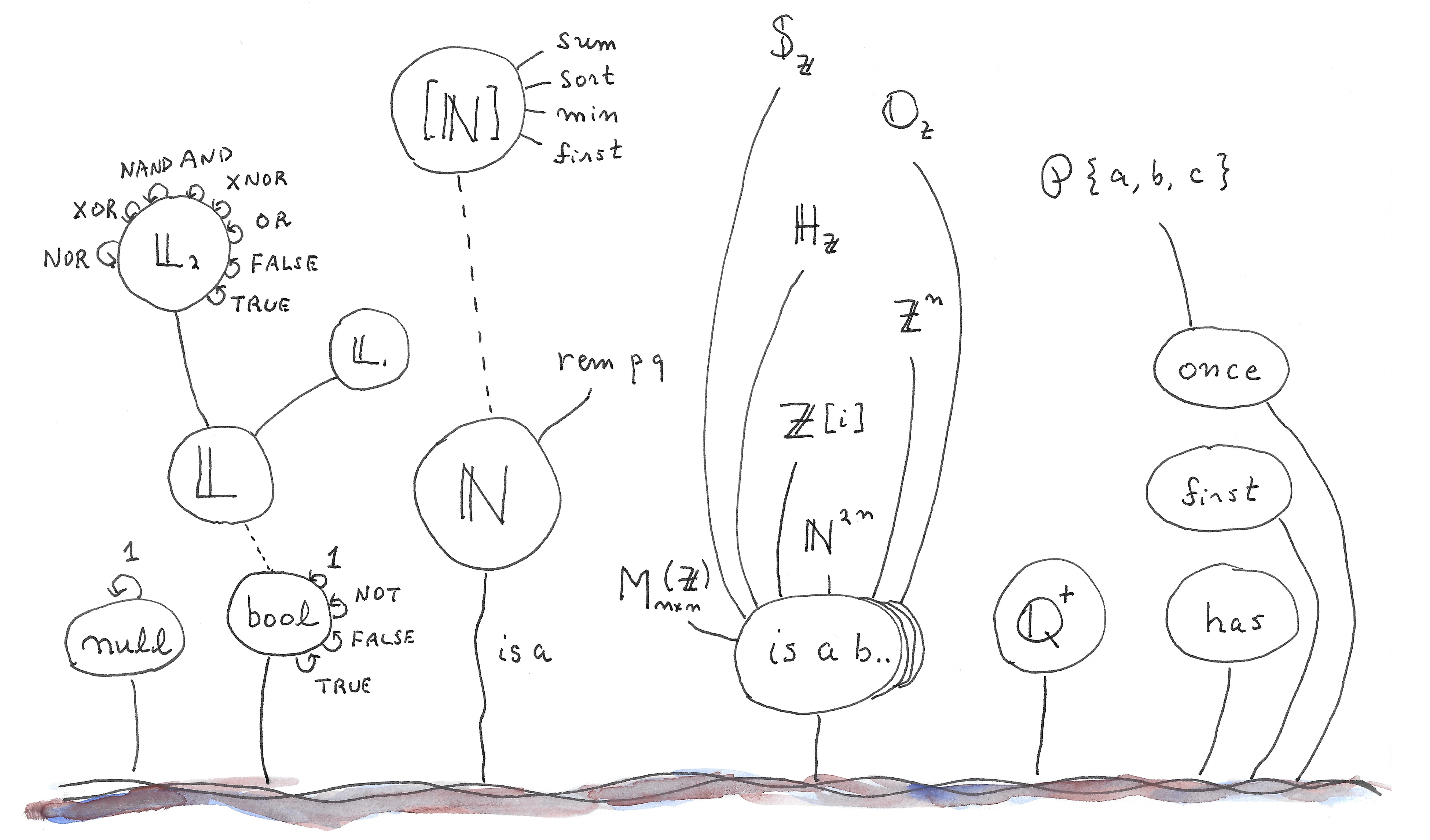}
\caption{{\it Spaces and subspaces discussed in the text, viewed as emerging ``organically'' from (pass), the space of all pure data. 
Organic numbers are subspaces of (is a..) are shown in the center with bool and descendants on the left, sets and the bespoke
${\mathbb Q}^+$ construction to the right of the figure.  Dashed lines indicate the action of Seq functors.}}
\end{figure}

\section{Discussion} 

      Within the axiomatic framework of pure data, one needs a concept of collections of mathematical objects analogous to sets in Set Theory or objects  
in Category Theory.   We find that data representing a collection of other data needs to satisfy an intuitive and simple condition.  Such data must 
be a {\it space}.  The morphisms corresponding to spaces are then just products of data which start and end with a space.  
A space, together with the endomorphisms of the space are a {\it semiring} which may be {\it semialgebras}, {\it central semialgebras} or 
{\it fields}.  This surprisingly rich structure is shared by all collections ranging from a single atom (:) to the entire space of all pure data as shown 
in Figure 3.  

     After working out the basic theory, we develop a set of examples, partly to illustrate the theory, and partly as an exploration of the general 
space of mathematical results from the pure data perspective.  We adopt the theme of letting spaces emerge ``organically'' from the simplest 
combinatoric definitions which are inevitable from the foundational finite sequence concept.  
A second major theme is to focus on algebraic data because such data, in some sense, transcends 
the foundational concept, earning a platonic meaning independent of the pure data framework.  A more 
sophisticated version of this idea is to focus on algebraic spaces and on central homomorphisms as the most likely objects 
of mathematical interest. 
Figure 8 illustrates the spaces explored with these considerations in mind.  
We find that the purely organic approach is at least partially successful.  A substantial fraction of natural combinatoric definitions are 
spaces already, and, for instance, the {\it organic numbers} show that a substantial fraction of the resulting spaces are important classical algebras.  
The rational numbers, as described in Section 11, is an example of a more intentional approach where one has something mathematical
in mind from the beginning.  This process is interestingly different from what one does in classical mathematics.  The rational numbers 
must be a space, since, basically, everything must be a space.  This means that rather than inventing axioms, one starts with a space 
large enough to encompass what one wants and then trys to coax it towards the desired space by passing into subspaces.  In this 
process one is only defining the main addition operation of a space, hoping or trusting that the resulting homomorphisms will end up 
doing what you intended from the beginning.  This seems much more constrained than standard mathematics.  On the other hand, 
one does seem to end up with the intended structure.  

\subsection{Open Questions}

These results raise a number of questions which might be answered by further research.  

\subsubsection{How can we find more ``organic'' spaces?}

One advantage of the pure data approach is that spaces are ``born with computation'' and the particular space (pass), by definition,  
contains {\it all} mathematical objects of any kind\cite{coda}.  Naively, we could search the space of pure data, say, up to 
a specified depth and width and check each data to see if it is a space, thus producing a list of {\it all} mathematical objects up to 
some specified level of simplicity.  
\begin{table}[h!]
\centering
\begin{tabular}{c|ccrrr}
\textbf{width} $\backslash$ \textbf{depth} & 0 & 1 & 2 & 3 & 4 \\
\hline
0 & 1 & 1 & 1 & 1 & 1 \\
1 & 1 & 2 & 5 & 26 & 677 \\
2 & 1 & 3 & 91 & $6.9\times 10^{7}$ & $2.2\times 10^{31}$ \\
3 & 1 & 4 & 4369 & $7.0\times 10^{21}$ & $1.1\times 10^{131}$ \\
4 & 1 & 5 & 406901 & $7.5\times 10^{44}$ & $1.0\times 10^{359}$ \\
\end{tabular}
\caption{{\it The number of pure data grows rapidly with maximum width and depth.}}
\label{tab:pure_data_counts}
\end{table}
This is actually how we started searching for ``organic'' spaces as illustrated in Figure 2.  Unfortunately, 
Table 3 shows how drastically impractical it is to extend a naive approach, even slightly.  
There are clearly strategies to do better than a brute force approach, but these are as yet unexplored.  

\subsubsection{Is there a missing theory of morphisms?}

The theory of section 6 is almost exclusively an ``inward facing'' look at the endomorphisms of one space at a time.  Other than 
subspaces, ``outward facing'' morphisms between spaces have barely been touched.  We have done little to relate spaces with 
similar properties to each other.  The situation may be analogous to understanding Set theory but not yet having a theory of categories 
where morphisms between similar spaces are understood in some coherent way.  One can expect a general theory of morphisms to exist, 
but it should also be substantially different than classical Category Theory, and so could be a source of new insights.

\subsubsection{How is the rest of Mathematics supposed to work?}

In \cite{PDF} and in this paper, we presume that mathematics in general fits into the overall framework of pure data.  
But we also argue that anything that's at least a collection of data must be a space, which automatically comes
with a semiring.  This appears to work well for the examples shown here.  The ``organically grown'' 
spaces, the ``bespoke rationals'' and the spaces resulting from Seq functors are all well described as semirings and 
all have a similar algebraic flavor.  However, it is not clear how well this will work for the parts of Mathematics 
which are far from the theory of semirings.

\subsubsection{Is there a deeper theory of spaces?}

The theory of spaces developed in Section 6 is mainly a collection of definitions and the most simple results.  This is only scratching 
the surface of what seems like a rich subject which is likely related to several areas of known mathematics 
\cite{golan1999,kuich1986,howie1995,semigroup,eilenberg1976,baader1998,dershowitz1990}.

\section{Glossary}
Some basic definitions used in the text and associated software\cite{coda}.  Lower case letters indicate atoms. 

\begin{enumerate}
\item {Basics
\begin{enumerate}
\item{$\delta_{\rm const}$ : (const A : B) $\mapsto$ A}
\item{$\delta_{\rm put}$ : (put A : B) $\mapsto$ (A:B)}
\item{$\delta_{\rm get}$ : (get A:B) $\mapsto$ ap (get0 A:B)}  
\item{$\delta_{\rm get0}$ : (get A: (A:B)) $\mapsto$ B}
\item{$\delta_{\rm atoms}$ : (atoms : b B) $\mapsto$ (:)\ (atoms : B)}
\end{enumerate}
}
\item {Control
\begin{enumerate}
\item{$\delta_{\rm if}$: (if ():B) $\mapsto$ B }
\item{$\delta_{\rm if}$: (if a A:B) $\mapsto$ () }
\item{$\delta_{\rm nif}$: (nif ():B) $\mapsto$ () }
\item{$\delta_{\rm nif}$: (nif a A:B) $\mapsto$ B }
\item{$\delta_{\rm while}$: (while A:B) $\mapsto$ B if (A:B)=B }
\item{$\delta_{\rm while}$: (while A:B) $\mapsto$ (while A: A : B) }
\end{enumerate}
}
\item{Semiring
\begin{enumerate}
\item{$\delta_{\rm prod}$: (prod a A:B) $\mapsto$ (a$^R$: prod A:B)}
\item{$\delta_{\rm prod}$: (prod :B) $\mapsto$ B}
\item{$\delta_{\rm sum}$: (sum a A:B) $\mapsto$ (a$^R$:B) (sum A:B)}
\item{$\delta_{\rm sum}$: (sum : B) $\mapsto$ () }
\end{enumerate}
}
\item{Definition}
\begin{enumerate}
\item{$\delta_{\rm n}$: (n A:B) $\mapsto$ (n A:B).  Defines n-atoms as fixed points with {\bf domain} n.} 
\item{$\delta_{\rm domain}$ : (domain : b B) $\mapsto$ ({\it domain of the atom b}) (domain:B)}
\end{enumerate}
\item{Combinatorics} 
\begin{enumerate}
\item{$\delta_{\rm map}$ : (map A:B) $\mapsto$ aq \{if ((arg:A):B):(right:B):B\} A:B}
\item{$\delta_{\rm aq}$: (aq a1 a2 A:B) $\mapsto$ (a1 a2:B) (aq a1 A:B)}
\item{$\delta_{\rm aq}$: (aq :B) $\mapsto$ ()}
\item{$\delta_{\rm aq}$: (aq A:) $\mapsto$ ()}
\item{$\delta_{\rm ar}$: (ar a A:B) $\mapsto$ (a a$_i$ : b$_j$) for $a_i$ in A and b$_j$ in B.} 
\end{enumerate}
\item{Sequence 
\begin{enumerate}
\item{$\delta_{\rm first}$ : (first : b B) $\mapsto$ b}
\item{$\delta_{\rm first}$ : (first : ()) $\mapsto$ ()} 
\item{$\delta_{\rm last}$ : (last : B b) $\mapsto$ b} 
\item{$\delta_{\rm last}$ : (last : ()) $\mapsto$ ()} 
\item{$\delta_{\rm has}$ : (has A:B) $\mapsto$ ap \{if (A=(domain:B)):B\} A : B}
\item{$\delta_{\rm hasnt}$ : (hasnt A:B) $\mapsto$ ap \{nif (A=(domain:B)):B\} A : B}
\item{$\delta_{\rm is}$ : (is A:B) $\mapsto$ ap \{nif\ (ar \{not:A=B\} A:B):B\} }
\item{$\delta_{\rm isnt}$ : (is A:B) $\mapsto$ ap \{if\ (ar \{not:A=B\} A:B):B\} }
\item{$\delta_{\rm once}$ : (once A:b B) $\mapsto$ (once A:b) (once A:B)} 
\item{$\delta_{\rm once}$ : (once A:b) $\mapsto$ b, if b$\in$A}
\item{$\delta_{\rm rev}$ : (rev : B b) $\mapsto$ b (rev:B)} 
\item{$\delta_{\rm rev}$ : (rev : () ) $\mapsto$ ()}
\item{$\delta_{\rm remove}$ : (remove A : A B) $\mapsto$ B} 
\item{$\delta_{\rm remove}$ : (remove A : B) $\mapsto$ B, if B does not start with A.} 
\item{$\delta_{\rm sort}$ : (sort A: B) $\mapsto$ B, sorted in lexical order.} 
\end{enumerate} 
}
\end{enumerate}

%%%%%%%%%%%%%%%%%%%%%%%%%%%%%%%%%%%%%%%%%%%%%%%%%%%%%%%%%%%%%%%%%%%%%%%%
%%% 

%%%%%%%%%%%%%%%%%%%%%%%%%%%%%%%%%%%%%%%%%%%%%%%%%%%%%%%%%%%%%%%%%%%%%%%%
\end{document}